\documentclass[journal]{IEEEtran}

\usepackage{amsfonts}
\usepackage{amsmath}
\usepackage{graphicx}
\usepackage{url}
\usepackage{eepic, epsfig, amsmath, amssymb, latexsym, setspace, subfig}
\usepackage{rotating}
\usepackage{amsfonts}
\usepackage{amsmath}
\usepackage{graphicx}
\usepackage{nicefrac}
\usepackage{lipsum,multicol}

\newtheorem{theorem}{Theorem}[section]

\newtheorem{lemma}[theorem]{Lemma}

\numberwithin{equation}{section}

\newcommand{\qed}{\rule{7pt}{7pt}}
\newenvironment{proof}{\noindent{\bf Proof}\hspace*{1em}}{\hfill\qed\vspace{0.125in}}

\begin{document}

\title{Sensing with Optimal Matrices}

\author{ Hema Kumari Achanta, Soura Dasgupta, and Weiyu Xu \\
Department of Electrical and Computer Engineering \\
University of Iowa\thanks{The authors are temporarily listed in alphabetical order of their names.}}

\maketitle

\begin{abstract}
We consider the problem of designing optimal $M \times N$ ($M \leq  N$) sensing matrices which minimize the maximum condition number of all the submatrices of $K$ columns. Such matrices minimize the worst-case estimation errors when only $K$ sensors out of $N$ sensors are available for sensing at a given time.  For $M=2$ and matrices with unit-normed columns, this problem is equivalent to the problem of maximizing the minimum singular value among all the submatrices of $K$ columns. For $M=2$, we are able to give a closed form formula for the condition number of the submatrices.  When $M=2$ and $K=3$, for an arbitrary $N\geq3$, we derive the optimal matrices which minimize the maximum condition number of all the submatrices of $K$ columns. Surprisingly, a uniformly distributed design is often \emph{not} the optimal design minimizing the maximum condition number.

\end{abstract}

\IEEEpeerreviewmaketitle

\section{Introduction}
Consider a set of $N$ sensors. These $N$ sensors are used to estimate an $M$-dimensional signal, where $N\geq M$. In the sensor scheduling problems, to maximize the lifetime of the sensor network, at any single time instant, only $K$ sensors are turned on to monitor the $M$-dimensional signal. In our system, we assume that each time these $K$ sensors are uniformly selected from the $\binom{N}{K}$ possible subsets, so on average the lifetime of the sensor network is extended by a factor of $\frac{N}{K}$. In hostile environments, for example, in battlefields, it is very common that only a limited number of sensors are able to survive and operate as designed. Suppose that only $K$ sensors out of the $N$ sensors are able to survive the hostile environment and be functional in sensing the $M$-dimensional signal. For these application scenarios, it is helpful to maximize the worst-case performance of the sensing system, no matter what set of sensors are used or are able to survive. In this paper, we consider the problem of optimal sensing schemes to achieve such a goal.

Suppose the signal is denoted as a vector $x \in R^{M}$. Let us consider a sensing matrix $A \in R^{M\times N}$. The sensing results of $N$ sensors can be represented by $N$ real numbers, each of which is the inner product between the signal $x$ and a column of $A$. Let ${KS} \subseteq \{1,2..., N\}$, with cardinality $|KS|=k$, be the subset sensors that are active at a certain time. We can then represent the measurement matrix of the surviving sensors by $A_{{KS}}$, where $A_{KS}$ is a $M \times K$ matrix consisting of columns indexed by $KS$ from $A$.

Then the $K$-dimensional measurement result $y$
\begin{equation*}
y=A_{KS}^{T}x+w,
\end{equation*}
where $w$ is the $K$-dimensional observation noise.

Then by the singular value decomposition, we have
\begin{equation*}
A_{KS}^{T}=U\Lambda V^{*},
\end{equation*}
where $\Lambda$ is a $M \times M$ diagonal matrix with the $K$ singular values $\sigma_1$, $\sigma_2$, $...$, and $\sigma_M$ on its diagonal.

Then the estimation error of $\hat{x}=(A_{KS}^{T}A_{KS})^{-1}A_{KS}^{T}(y)$ satisfies
\begin{equation*}
\|\hat{x}-x\|_2=\|(A_{KS}^{T}A_{KS})^{-1}A_{KS}^{T}(w)\|_2 \leq \frac{\|w\|_2}{\sigma_{min}}.
\end{equation*}

So in order to optimize the worst-case performance, we need to maximize the smallest singular value among all the possible $\binom{N}{K}$ possible subsets. This introduces a problem of designing the measurement matrix $A$. To make the problem meaningful, we assume that each column of the measurement matrix has unit norm.  Since each column of the measurement matrix $A$ has unit norm, when $M=2$, this is equivalent to minimizing the maximum condition number. 

In general, the condition number $\kappa(B)$ of a matrix $B$ is the ratio of the largest singular value $\sigma_{max}(B)$ and the smallest singular value $\sigma_{min}(B)$: $\kappa(B)=\frac{\sigma_{max}(B)}{\sigma_{min}(B)}$.

Let $A=[a_{1}, a_{2}, ...,  a_{N}]$,  where $a_{1},....,a_{N}$ are the columns of $A$. We assume here that $||a_{i}||_2=1$ holds true for $1\leq i \leq N$. Let ${KS} \subseteq \{1,2,...,N\}$ with cardinality $|KS|=K$. Now, let $A_{{KS}}$ be an $M \times K$ submatrix $A_{{KS}}=[a_{i_{1}}, a_{i_{2}},.....,a_{i_{K}}]$ with columns indices $i_j$, $1\leq j \leq K$, from the set $KS$. Also define
\begin{equation}
 \tilde{A}_{KS} = A_{{KS}}A_{{KS}}^T = \sum_{j=1}^{K} a_{i_{j}}a_{i_{j}}^T
\end{equation}
Using these notations, we can describe our optimal design problem for the parameter set $(M,N, K)$ as follows.
\begin{equation*}
\min_{A \in R^{M \times N} \text{with unit-normed columns}}  \left\{ \underset{{KS\subseteq \{1,2,...,N\}}}{\text{max}} \frac{\lambda_{max}(\tilde{A}_{KS})}{\lambda_{min}(\tilde{A}_{KS})}\right\}.
\end{equation*}

Compared with the design of compressive sensing matrices satisfying the restricted isometry condition \cite{CS}, in our problem, the submatrices $A_{KS}$ are wide matrices instead of tall matrices in \cite{CS}. Also, the application background is very different from compressive sensing. 

\section{Derivation of the Condition Number for $M=2$}
Generally, the optimal design for an arbitrary $M$, $N$ and $K$ is difficult to get. The difficulty arises from the fact, we need to optimize the maximum condition number among $\binom{N}{K}$ submatrices. In our applications, we focus on the case of $M=2$. When $M=2$, we can a concise formula for the condition number for a specific submatrix $\tilde{A}_{KS}$. 

We know that the condition number of $\tilde{A}_{KS}$ is given by
\begin{equation}
\kappa(\tilde{A}_{KS}) =\frac{\max_{||\eta||=1}(\eta^{T}\tilde{A}_{KS}\eta)}{\min_{||\eta||=1}(\eta^{T}\tilde{A}_{KS}\eta)}
\end{equation}
Since the columns of $A$ are unit-normed, we can represent $A=[a_1, a_2, ...., a_N]$ with
\begin{equation}
a_{i}=\left( \begin{array}{c}
\cos \theta_{i} \\
\sin \theta_{i}\\\end{array} \right)
\end{equation}
for $1\leq i \leq N$, where $\theta_{i}\in [0,\pi)$. Note that shifting $\theta_i$ by $\pi$ does not change the condition number of any submatrix.

Since $||\eta||_2=1$  we can choose
\begin{equation*}
\eta =\left( \begin{array}{c}
\cos \alpha \\
\sin \alpha\\\end{array} \right).
\end{equation*}

Thus
\begin{equation}
 \eta^{T} \tilde{A}_{KS} \eta = \{\sum_{j=1}^{{K}} \eta^{T}  a_{i_j}a_{i_j}^T\eta\}.
 \label{eq:sumproduct}
\end{equation}
And, $\eta^{T}  a_{i_j}a_{i_j}^T\eta$ is equal to
\begin{eqnarray*}
&& (\cos(\alpha)\cos(\beta_{i_j})+\sin(\alpha)\sin(\beta_{i_j}))^2\\
&&= \cos^{2}(\alpha-\theta_{i_j})
\end{eqnarray*}

After simplification, (\ref{eq:sumproduct}) becomes
\begin{equation}
 \eta^{T} \tilde{A}_{KS} \eta = \sum_{j=1}^{{K}}\cos^{2}(\alpha-\theta_{i_j}) = \frac{{K}}{2}+ \frac{1}{2}\sum_{j=1}^{{K}}\cos(2(\alpha-\theta_{i_j})).
\end{equation}
Let us define
\begin{equation}
 J(\alpha) =\frac{{K}}{2}+ \frac{1}{2}\sum_{j=1}^{{K}}\cos(2(\alpha-\theta_{i_j})).
 \label{eq:Jdefinition}
\end{equation}
Then the minimum or maximum eigenvalue of $\tilde{A}_{KS}$ is achieved when $J'(\alpha)=0$
\begin{equation}
 J'(\alpha) = -2\sum_{j=1}^{{K}}\sin(2(\alpha - \theta_{i_j}))=0.
 \label{eq:derivative0}
\end{equation}
We also have 
\begin{equation}
 J''(\alpha) = -4\sum_{j=1}^{{K}}\cos(2(\alpha - \theta_{i_j})) \leq 0
\end{equation}
at the maximum eigenvalue and the inequality is reversed at the minimum eigenvalue. An important observation to make is that $\alpha_{max}$ and $\alpha_{min}$ differ by $\frac{\pi}{2}$.

When ${(\sum_{j=1}^{{K}}\sin(2\theta_{i_j}))^2+(\sum_{j=1}^{{K}}\cos(2\theta_{i_j}))^2} \neq 0$, from (\ref{eq:derivative0}), the optimizing $\alpha$ satisfies
\begin{equation*}
 \cos(2\alpha)= \frac{\sum_{j=1}^{{K}}\cos(2\theta_{i_j})}{ \sqrt{{(\sum_{j=1}^{{K}}\sin(2\theta_{i_j}))^2+(\sum_{j=1}^{{K}}\cos(2\theta_{i_j}))^2}}}
\end{equation*}
and
\begin{equation*}
  \sin(2\alpha)= \frac{\sum_{j=1}^{{K}}\sin(2\theta_{i_j})}{\sqrt{{(\sum_{j=1}^{{K}}\sin(2\theta_{i_j}))^2+(\sum_{j=1}^{{K}}\cos(2\theta_{i_j}))^2}}}.
\end{equation*}
From expansion of (\ref{eq:Jdefinition}), we get
\begin{equation}
 J(\alpha) =\frac{{K}}{2}+ \frac{1}{2}\sum_{j=1}^{{K}}\cos(2\alpha)\cos(2\theta_{i_j})+\frac{1}{2}\sum_{j=1}^{{K}}\sin(2\alpha)\sin(2\theta_{i_j}).
 \label{eq:Jexpansion}
\end{equation}
Combining the optimizing $\alpha$ and (\ref{eq:Jexpansion}), we have
\begin{eqnarray}
&&J(\alpha) = \frac{{K}}{2} +\nonumber\\
&&\frac{1}{2}
\frac{\sum_{j=1}^{{K}}\sum_{l=1}^{{K}}(\cos(2\theta_{i_l})\cos(2\theta_{i_j})+ \sin(2\theta_{i_l})\sin(2\theta_{i_j})) }
{\sqrt{(\sum_{l=1}^{{K}}\sin(2(\theta_{i_l})))^2+(\sum_{l=1}^{{K}}\cos(2\theta_{i_l}))^2}}.\nonumber
\label{eq:Jexpansion2}
\end{eqnarray}
Define $\mathrm{den}^2=(\sum_{l=1}^{{K}}\sin 2\theta_{i_l})^2+(\sum_{l=1}^{{K}}\cos 2\theta_{i_l})^2$.  Then
\begin{eqnarray*}
\mathrm{den}^2&=&\sum_{l=1}^{{K}}(\sin^2 2\theta_{i_l}+\cos^2 2\theta_{i_l})\nonumber\\
&+& \sum_{j=1}^{{K}}\sum_{l=1,l \neq j}^{{K}} \cos2\theta_{i_l}\cos2\theta_{i_j}+\sum_{j=1}^{{K}}\sum_{l=1}^{{K}} \sin 2\theta_{i_l}\sin 2\theta_{i_j}\nonumber\\
&=&{K} + 2\sum_{j=1}^{{K}}\sum_{l=j+1}^{{K}} \cos2(\theta_{i_l}-\theta_{i_j})\nonumber
\end{eqnarray*}
Similarly, we define $\mathrm{num} = \sum_{j=1}^{{K}} \sum_{l=1}^{{K}}\cos(2\theta_{i_l})\cos(2\theta_{i_j})+\sum_{j=1}^{{K}} \sum_{l=1}^{{K}}\sin(2\theta_{i_l})\sin(2\theta_{i_j})$. 

It can be expanded as
\begin{IEEEeqnarray*}{rcl}
\mathrm{num}&=& \sum_{l=1}^{{K}}(\sin^2 2\theta_{i_l}+\cos^2 2\theta_{i_l})+ \sum_{j=1}^{{K}}\sum_{l=1,l \neq j}^{{K}} \cos2\theta_{i_l}\cos2\theta_{i_j}\IEEEnonumber\\
&+&  \sum_{j=1}^{{K}}\sum_{l=1,l \neq j}^{{K}} \sin 2\theta_{i_l}\sin 2\theta_{i_j}\IEEEnonumber\\
&=&{K} + 2\sum_{j=1}^{{K}}\sum_{l=j+1}^{{K}} \cos2(\theta_{i_l}-\theta_{i_j})\nonumber
\end{IEEEeqnarray*}

Plugging $\mathrm{den}$ and $\mathrm{num}$ into (\ref{eq:Jexpansion2}), we get
\begin{equation}
J(\alpha_{max})=\frac{{K}}{2}+ \frac{1}{2}\sqrt{\frac{{K}}{2}+\sum_{j=1}^{{K}}\sum_{l=j+1}^{{K}} \cos2(\theta_{i_l}-\theta_{i_j})},
\label{eq:Jmax}
\end{equation}
and \begin{equation}
J(\alpha_{min})=\frac{{K}}{2}- \frac{1}{2}\sqrt{\frac{{K}}{2}+\sum_{j=1}^{{K}}\sum_{l=j+1}^{{K}} \cos2(\theta_{i_l}-\theta_{i_j})}.
\label{eq:Jmin}
\end{equation}
Thus minimizing the condition number of $\tilde{A}_{{KS}}$ for a given set of indices $\{i_{1},i_{2},..,i_{{K}}\}$ is the same as this optimization problem 
 \begin{equation}
\begin{aligned}
& \underset{}{\text{minimize}}
& &  {\sum_{j=1}^{{K}}\sum_{l=j+1}^{{K}} \cos2(\theta_{i_l}-\theta_{i_j})}.
\end{aligned}
\end{equation}

With $KS\subseteq \{1,2,...,N\}$,  the optimal sending matrix design problem for $M=2$ can be reformulated as,
\begin{eqnarray*}
\min_{\theta_{1}, ..., \theta_{N}}  \underset{KS=\{i_{1},i_{2},..,i_{K}\} } {\text{max}}
{\sum_{j=1}^{{K}}\sum_{l=j+1}^{{K}} \cos2(\theta_{i_l}-\theta_{i_j})}.\\
\end{eqnarray*}

One can easily find the optimal solution for $K=2$.
\begin{theorem}
Let ${K}=2$, $M=2$ and let $N\geq 2$ be an integer. Then the set of angles $\Theta=\{0,\frac{\pi}{N},\frac{2\pi}{N}...\frac{(N-1)\pi}{N}\}$ minimizes the maximum condition number over all possible sub-matrices with two columns.
\end{theorem}
\begin{proof}
The optimal design  minimizes the cost function
\begin{equation}
\begin{aligned}
& \underset{}{}
& &  {\cos2(\theta_{i_l}-\theta_{i_j})} \\
\end{aligned}
\end{equation}
for the set of indices $\{i_{1},i_{2}\} \subseteq \{1,2,...,N\}$ which gives the largest cost function.

Without loss of generality, we let $\theta_{i}$, $1\leq i \leq N$, lie in the range $[0,\pi)$ and let $\theta_{1}=0$. In order to minimize the maximum condition number, we only need to maximize the minimum of $\min\{|2\theta_{i_l}-2\theta_{i_j}|, 2\pi-(2\theta_{i_l}-2\theta_{i_j})|\}$. This is apparently achieved with the given set of angles.
\end{proof}

In the following sections, we will derive the optimal design for $K=3$.

\section{$K=3$, $N$ is an even number}

Surprisingly, unlike $K=2$, the optimal matrix design for $K=3$ is often not achieved with the uniformly distributed angles. 

\begin{theorem}
Let $K=3$ and $N$ be an even number. Then the set of angles $\theta_i=\frac{2\pi (i-1)}{N} \mod \pi$, $1 \leq i \leq  N$, minimizes the maximum condition number among all sub-matrices with $K$ columns. Moreover, they are the unique set of angles that achieve the smallest maximum condition number for $N \geq 6$.
\label{thm:K3NEven}
\end{theorem}

\begin{proof}
We first derive a lower bound for the maximum condition number among all sub-matrices with $K=3$ columns; and then show the given set of angles achieve this lower bound.

Suppose that the set of angles  $0\leq \theta_i^*<\pi$, $1\leq i \leq N$,  achieve the smallest maximum condition number for all submatrices with $K=3$ columns. Without loss of generality, let $\theta_1^*=0$; and let $\theta_i^*$, $1\leq i \leq N$, appear sequentially in a counter-clockwise order. Let $\tilde{\theta}_i=2\theta_{i}^*$, so we have $0\leq \tilde{\theta}_i <2\pi$.

  \emph{(lower bound for maximum condition number)}\\
  We claim that there must exist an index $1\leq i \leq N$ such that for $\tilde{\theta}_{i}$, $\tilde{\theta}_{(i+1)\mod N}$, and $\tilde{\theta}_{(i+2)\mod N}$,  $|(\tilde{\theta}_{(i+2)\mod N}-\tilde{\theta}_{i})\mod(2\pi)| \leq \frac{4\pi}{N}$.  Notice that $|(\tilde{\theta}_{(i+2)\mod N}-\tilde{\theta}_{i})\mod(2\pi)|$ is just the counter-clockwise region going from $\tilde{\theta}_{i}$ to $\tilde{\theta}_{(i+2)\mod N}$. So the summation $ \sum_{i=1}^{N}|(\tilde{\theta}_{(i+2)\mod N}-\tilde{\theta}_{i})\mod(2\pi)| =2 \times (2\pi)$ because each counter-clockwise region between two adjacent angles are summed twice. By looking at the average of these $N$ summands,  such an index $i$ must exist.

 For simplicity of notations, we denote these three angles $\tilde{\theta}_{i}$, $\tilde{\theta}_{(i+1)\mod N}$, and $\tilde{\theta}_{(i+2)\mod N}$ as $t_1$, $t_2$ and $t_3$.  Without loss of generality, we assume that $0=t_1 \leq t_2 \leq t_3 \leq \frac{4\pi}{N}$.  We how that the smallest condition number that these three angles $t_1$, $t_2$, and $t_3$ can achieve is when $t_2=t_1$ or $t_2=t_3$.

  We consider the scenario where $|t_3-t_1| \leq \frac{4\pi}{N}$ remains as a fixed constant. Define $f(t_2)$ as
  \begin{equation*}
  f(t_2)=\cos(t_1-t_2)+\cos(t_1-t_3)+\cos(t_2-t_3).
  \end{equation*}
Its derivative is
  \begin{eqnarray*}
  f'(t_2)&=&-\sin(t_2-t_1)+\sin(t_3-t_2)\\
  &=&2\sin(\frac{t_3+t_1}{2}-t_2)\cos(\frac{t_3-t_1}{2}).
  \end{eqnarray*}

So if $(t_3-t_1) \leq \pi$, the derivative $f'(t_2)$ is non-positive for $\frac{t_3+t_1}{2}\leq t_2 \leq \frac{t_3+t_1}{2}+\pi$; and it is non-negative for $\frac{t_3+t_1}{2}+\pi\leq t_2 \leq \frac{t_3+t_1}{2}+2\pi$. So if $0=t_1 \leq t_2 \leq t_3 \leq \frac{4\pi}{N}$, $f(t_2)$ is minimized when $t_2=t_1$ or $t_2=t_3$. The corresponding $f(t_2)$ is
\begin{equation*}
f(t_2=t_1)=f(t_2=t_3)=1+2\cos(t_1-t_3) \geq 1+2\cos(\frac{4\pi}{N}).
\end{equation*}

\emph{(Achievability)}\\
 In order to finish the proof, we only need to show that the given set of angles $\theta_i=\frac{2\pi (i-1)}{N} \mod \pi$, $1 \leq i \leq  N$, achieve the lower bound $1+2\cos(\frac{4\pi}{N})$. Let $\ddot{\theta}_i=2\theta_{i}$, so we have $0\leq \ddot{\theta}_i <2\pi$. Counter-clockwise, starting from the two angles $\ddot{\theta}_i=0$ and $\ddot{\theta}_{\frac{N}{2}+1}=0$ (which are in fact two angles in the same position), we \emph{re-label} these $N$ angles sequentially as $\hat{\theta}_1$, $\hat{\theta}_2$, ..., and $\hat{\theta}_{N}$.

 Namely, we need to show, for any $3$ angles $r_1$, $r_2$ and $r_3$ from the given set of angles $\hat{\theta}_i$ satisfy
\begin{equation*}
\cos(r_1-r_2)+\cos(r_2-r_3)+\cos(r_1-r_3) \leq 1+2\cos(\frac{4\pi}{N}).
\end{equation*}

Without loss of generality, we assume that $r_1$, $r_2$ and $r_3$ are in a counter-clockwise order; and assume that $|(r_2-r_1)\mod(2\pi)|$ is the smallest among $|(r_2-r_1)\mod (2\pi)|$, $|(r_3-r_2) \mod (2\pi)|$ and $|(r_1-r_3)\mod (2\pi)|$. Apparently, $|(r_2-r_1)\mod(2\pi)| \leq \frac{2\pi}{3}$, and $|(r_2-r_1)\mod(2\pi)|$ is an integer multiple of $\frac{4\pi}{N}$.

Suppose $|(r_2-r_1)\mod(2\pi)|=0$. Then $r_2=r_1$ and $|(r_1-r_3)\mod(2\pi)|=|(r_3-r_2)\mod(2\pi)| \geq \frac{4\pi}{N}$. Similar to the proof of ``lower bound'', for such a setting, the function
  \begin{equation*}
  f(r_3)=\cos(r_1-r_2)+\cos(r_1-r_3)+\cos(r_2-r_3)
  \end{equation*}
is a decreasing function of $r_3$ for $r_3 \in [r_1, (r_1+\pi) \mod (2\pi)]$; and an increasing function of $r_3$ for $r_3 \in [(r_1+\pi) \mod (2\pi), (r_1+2\pi) \mod (2\pi)]$. So the maximum of $f(r_3)$ is achieved when $|(r_1-r_3) \mod (2\pi)|=\frac{4\pi}{N}$, where $f(r_3)=1+2\cos(\frac{4\pi}{N})$.

Suppose $|(r_2-r_1)\mod(2\pi)|=\frac{4\pi}{N}$. Then $|(r_1-r_3)\mod(2\pi)| \geq \frac{4\pi}{N}$ and $|(r_3-r_2)\mod(2\pi)| \geq \frac{4\pi}{N}$. Similar to the reasoning in the ``lower bound'' part, the maximum for $f(r_3)=\cos(r_1-r_2)+\cos(r_1-r_3)+\cos(r_2-r_3)$ is achieved when $|(r_3-r_2) \mod (2\pi)|=\frac{4\pi}{N}$, $r_3 \neq r_1$; or $|(r_1-r_3) \mod (2\pi)|=\frac{4\pi}{N}$ and $r_3 \neq r_2$. In both cases, $f(r_3)=2\cos(\frac{4\pi}{N})+\cos(\frac{8\pi}{N})$, which is smaller than the lower bound $1+2\cos(\frac{4\pi}{N})$.

Now suppose $|(r_2-r_1)\mod(2\pi)|>\frac{4\pi}{N}$. Since $|(r_2-r_1)\mod(2\pi)| \leq \frac{2\pi}{3}$ and $r_3$ is certainly outside the counter-clockwise region going from $r_1$ to $r_2$, using the same reasoning in proving the lower bound, the function
  \begin{equation*}
  f(r_3)=\cos(r_1-r_2)+\cos(r_1-r_3)+\cos(r_2-r_3)
  \end{equation*}
achieves its maximum when $r_3=r_2$ or $r_3=r_1$. This maximum is
\begin{equation*}
  f(r_3=r_1)=1+2\cos(r_1-r_2),
\end{equation*}
which is certainly no bigger than $1+2\cos(\frac{4\pi}{N})$.

So the given set of angles indeed achieves the lower bound $1+2\cos(\frac{4\pi}{N})$, and we have proven the optimality of the given set of angles in minimizing the maximum condition number among all submatrices with $3$ columns.

\emph{(Uniqueness)}\\
Moreover, in the proof of the lower bound, when $N \geq 6$, $(t_3-t_1) < \pi$, the derivative $f'(t_2)$ is negative for $\frac{t_3+t_1}{2}< t_2 < \frac{t_3+t_1}{2}+\pi$; and positive for $\frac{t_3+t_1}{2}+\pi<t_2 < \frac{t_3+t_1}{2}+2\pi$. So $t_3=t_1$ or $t_3=t_2$ are the only two places where $f(t_3)$ achieves the lower bound $1+2\cos(\frac{4\pi}{N})$. We further notice that the lower bound is achieved only when the counter-clockwise region between any $3$ adjacent angles from $\tilde{\theta}_{i}$, $1\leq i \leq N$, is equal to $\frac{4\pi}{N}$. Otherwise, if there exist one set of $3$ adjacent angles from $\tilde{\theta}$ such that the region between them is larger than $\frac{4\pi}{N}$, there must exist another set of $3$ adjacent angles from $\tilde{\theta}_{i}$, $1\leq i \leq N$, such that the counter-clockwise region between them is smaller than $\frac{4\pi}{N}$. This is because $ \sum_{i=1}^{N}|\tilde{\theta}_{(i+2)\mod N}-\tilde{\theta}_{i}| =2 \times (2\pi)$. For these set of $3$ angles, their corresponding cost function $f(\cdot)$ is larger than the derived cost function lower bound $1+2\cos(\frac{4\pi}{N})$, thus bringing a larger maximum condition number. This proves for $N\geq6$, $\theta_i=\frac{2\pi (i-1)}{N} \mod \pi$, $1 \leq i \leq  N$ are the unique set of angles that minimize the maximum condition number.

\end{proof}

It is worth mentioning that when $N=4$, the design given in Theorem \ref{thm:K3NEven} is still optimal. However, we have more than one design that can minimize the maximum condition number.  This is because, when counter-clockwise region covered by $3$ angles is $\pi$, no matter where the middle angle is, the cost function is $-1$.
\begin{theorem}
For $N=4$, ${K}=3$ and $M=2$, the set of $\tilde \Theta =2\Theta= \{\tilde \theta_{1}, \tilde \theta_{2}, \tilde \theta_{1} + \pi, \tilde \theta_{2} + \pi\}$, where $0\leq \tilde \theta_{1}<\pi$, $0\leq \tilde \theta_{2}<\pi$, minimizes the maximum condition number over all possible $3 \times 3$ submatrices.
\end{theorem}

\section{$K=3$, $N=3$ or $5$}
Interestingly, when $K=3$, except for the trivial case $N=3$, $N=5$ is the only other case where a uniform distributed design indeed minimizes the maximum condition number.

\begin{theorem}
Let $K=3$ and $N=3$ or $5$. Then the set of angles $\theta_i=\frac{\pi (i-1)}{N}$, $1 \leq i \leq  N$, minimizes the maximum condition number among all sub-matrices with $K=3$ columns.
\label{thm:K3N5}
\end{theorem}

\begin{proof}
The case for $N=3$ is trivial, so now we only focus on proving the claim for $N=5$.

For the set of angles $\hat{\theta}_i=2\theta_i=\frac{2\pi (i-1)}{N}$, $1 \leq i \leq  N$, it is not hard to check that $3$ adjacent angles, denoted by $r_1$, $r_2$ and $r_3$, give the maximum cost function
\begin{equation*}
 \cos(r_1-r_2)+\cos(r_1-r_3)+\cos(r_2-r_3)=2\cos(\frac{2\pi}{5})+\cos(\frac{4\pi}{5}),
 \end{equation*}
which corresponds to largest condition number.

Let $0\leq \theta_{i}^*<\pi$, $1\leq i \leq N$, be a set of $N$ angles which minimizes the maximum condition number of all submatrices with $3$ columns. For convenience, we consider the corresponding $N$ angles $\tilde{\theta}_i=2\theta_{i}^*$, $1\leq i \leq N$. Without loss of generality, we assume $\tilde{\theta}_1=0$; and $0\leq \tilde{\theta}_i <2\pi$ are arranged sequentially in a counter-clockwise order as $i$ ranges from $1$ to $N$. We first prove the following two lemmas before proving that there must exist at least $4$ adjacent-3-angle sets which give the maximum condition number.

\begin{lemma}
The counter-clockwise region between any two adjacent angles (for example $\tilde{\theta}_i$ and $\tilde{\theta}_{(i+1)\mod{N}}$ for some $i$ ) is smaller than $\pi$.
\end{lemma}

\begin{proof}
Note that the $5$ angles partition the circle into $5$ regions. If instead  the counter-clockwise region going from $\tilde{\theta}_i$ to $\tilde{\theta}_{(i+1)\mod{N}}$ is at least $\pi$, because the other $4$ regions occupy at most $\pi$, there must exist three adjacent angles for which the two counter-clockwise regions covered by them is no bigger than $\frac{\pi}{2}$. For those three angles, from the same calculation as in Theorem \ref{thm:K3NEven}, the smallest cost function these $3$ angles can achieve is
\begin{equation*}
 \cos(\frac{\pi}{2}-\frac{\pi}{2})+\cos(0-\frac{\pi}{2})+\cos(0-\frac{\pi}{2})=1,
 \end{equation*}
which is already bigger than the cost $2\cos(\frac{2\pi}{5})+\cos(\frac{4\pi}{5})$ achieved by the $5$ angles $\tilde{\theta}_i=2\theta_i=\frac{2\pi (i-1)}{N}$, $1 \leq i \leq  N$.
\end{proof}

\begin{lemma}
Let $N=5$. In the optimal design $\tilde{\theta}_i$, $1\leq i \leq N$, consider $4$ adjacent angles $r_1$, $r_2$, $r_3$ and $r_4$, where they are arranged in a counter-clockwise order; and $r_2$ and $r_3$ are inside the counter-clockwise region going from $r_1$ to $r_4$. $(r_1,r_2,r_3)$ and $(r_2,r_3,r_4)$ give the same condition number if and only if $(r_2-r_1)\mod(2\pi)=(r_4-r_3)\mod(2\pi)$.
\end{lemma}

\begin{proof}
Without loss of generality, we assume $r_1=0$ such that $r_i$, $1\leq i \leq 4$, are all within $[0, 2\pi)$. Now we only need to show $(r_1,r_2,r_3)$ and $(r_2,r_3,r_4)$ give the same condition number if and only if $r_2-r_1=r_4-r_3$.

If $(r_1,r_2,r_3)$ and $(r_2,r_3,r_4)$ give the same condition number, we have
\begin{eqnarray*}
  &~&\cos(r_1-r_2)+\cos(r_1-r_3)+\cos(r_2-r_3)\\
  &=&\cos(r_3-r_2)+\cos(r_4-r_2)+\cos(r_4-r_3).
\end{eqnarray*}

This means
\begin{eqnarray*}
  &~&2\cos(\frac{r_3-r_2}{2}+r_2-r_1) \cos(\frac{r_3-r_2}{2})\\
  &=&2\cos(\frac{r_3-r_2}{2}+r_4-r_3) \cos(\frac{r_3-r_2}{2}).
\end{eqnarray*}

Since we have just shown that $r_3-r_2$ is smaller than $\pi$, we have
\begin{eqnarray*}
\cos(\frac{r_3-r_2}{2}+r_2-r_1)=\cos(\frac{r_3-r_2}{2}+r_4-r_3).
\end{eqnarray*}

This means either $r_2-r_1=r_4-r_3$ or $r_4-r_1=2\pi$. The latter is not possible for a set of angles which achieve the smallest maximum condition number, because $r_4-r_1=2\pi$ forces the next angle $r_5$ to be aligned with both $r_1$ and $r_4$. This gives a condition number of $\infty$ for the three angles $r_1$, $r_4$ and $r_5$. So we must have $r_2-r_1=r_4-r_3$.
\end{proof}

In the optimal design $\tilde{\theta}_i$, $1\leq i \leq N$, we assume that $\{\tilde{\theta}_1, \tilde{\theta}_2,  \tilde{\theta}_3\}$ is a adjacent-$3$-set which corresponds to the maximum condition number.

\subsection{(At Least $2$ adjacent-$3$-angle Sets Giving the Maximum Condition Number) }

\begin{lemma}
$\tilde{\theta}_1$, $\tilde{\theta}_2$ and $\tilde{\theta}_3$ can not be the unique set of $3$ angles that have the largest condition number.
\label{lem:5singleoptimal}
\end{lemma}

\begin{proof}
We prove by contradiction. Suppose $\{\tilde{\theta}_1, \tilde{\theta}_2, \tilde{\theta}_3\}$ is the unique set of $3$ adjacent angles that have the largest condition number. Then we must have $\tilde{\theta}_3-\tilde{\theta}_1<\pi$. Suppose instead that $\tilde{\theta}_3-\tilde{\theta}_1 =\pi$ or $\tilde{\theta}_3-\tilde{\theta}_1 >\pi$.

If $\tilde{\theta}_3-\tilde{\theta}_1 =\pi$, the cost function for the set of $3$ angles $\tilde{\theta}_1$, $\tilde{\theta}_2$ and $\tilde{\theta}_3$ is equal to the cost function for the set of $3$ angles $\tilde{\theta}_3$, $\tilde{\theta}_5$ and $\tilde{\theta}_1$. This is a contradiction to our assumption.

If $\tilde{\theta}_3-\tilde{\theta}_1 >\pi$, then $\tilde{\theta}_2-\tilde{\theta}_1=\tilde{\theta}_3-\tilde{\theta}_2$. Otherwise, we can always shift $\tilde{\theta}_2$ towards $\frac{\tilde{\theta}_1+\tilde{\theta}_3}{2}$ by a sufficiently small amount and strictly decrease the cost function for $\tilde{\theta}_1$, $\tilde{\theta}_2$ and $\tilde{\theta}_3$. Since the cost functions for any other $3$ adjacent angles are strictly smaller than the original cost function of $\tilde{\theta}_1$, $\tilde{\theta}_2$ and $\tilde{\theta}_3$, their cost functions will remain smaller than the new revised cost function of $\tilde{\theta}_1$, $\tilde{\theta}_2$ and $\tilde{\theta}_3$. So we have just decreased the largest condition number, which is a contradiction. So we must have $\tilde{\theta}_2-\tilde{\theta}_1=\tilde{\theta}_3-\tilde{\theta}_2$. However the cost function for $\tilde{\theta}_4$, $\tilde{\theta}_5$ and $\tilde{\theta}_1$ is lower bounded by $2\cos(\frac{\pi}{2})+\cos(\pi)=-1$, which is larger than the cost function for $\tilde{\theta}_1$, $\tilde{\theta}_2$ and $\tilde{\theta}_3$.  This is contradictory to the assumption that the set of $\tilde{\theta}_1$, $\tilde{\theta}_2$ and $\tilde{\theta}_3$ corresponds to the maximum condition number.

 So we must have $\tilde{\theta}_3-\tilde{\theta}_1<\pi$. In this case, if we shift $\tilde{\theta}_3$ counter-clockwise by a sufficiently small amount $\delta$, we will strictly decrease the cost function for $\tilde{\theta}_1$, $\tilde{\theta}_2$ and $\tilde{\theta}_3$. Since by our assumption, the cost function of any other $3$ adjacent angles were \emph{strictly} smaller than the original cost function of $ \{\tilde{\theta}_1,\tilde{\theta}_2,\tilde{\theta}_3\}$, their cost functions will stay smaller than the new revised cost function for $ \{\tilde{\theta}_1,\tilde{\theta}_2,\tilde{\theta}_3\}$. So we have just strictly decreased the maximum condition number of the optimal design, which is not possible.
\end{proof}

\subsection{(At Least $3$ adjacent-$3$-angle Sets Giving the Maximum Condition Number) }

\begin{lemma}
$\{\tilde{\theta}_1, \tilde{\theta}_2, \tilde{\theta}_3\}$ and $\{\tilde{\theta}_2, \tilde{\theta}_3, \tilde{\theta}_4\}$  can not be the only $2$ sets of $3$ adjacent angles which correspond to the maximum condition number.
\label{lemma:adjacenttwo}
\end{lemma}

\begin{proof}
We prove by contradiction. Suppose $\{\tilde{\theta}_1, \tilde{\theta}_2, \tilde{\theta}_3\}$ and $\{\tilde{\theta}_2, \tilde{\theta}_3, \tilde{\theta}_4\}$ are the only two sets of $3$ adjacent angles that have the largest condition number.
This means $\tilde{\theta}_4-\tilde{\theta}_3=\tilde{\theta}_2-\tilde{\theta}_1=\alpha$ for some $\alpha \geq 0$; and $\tilde{\theta}_2-\tilde{\theta}_1=\beta$ for some $\beta\geq 0$. Note that  $2\alpha+\beta<2\pi$ because, otherwise, $\tilde{\theta}_4$, $\tilde{\theta}_5$ and $\tilde{\theta}_1$ are forced to be in the same position, giving rise to a condition number of $\infty$ for these three angles.

From Lemma \ref{lemma:lessthanpi}, we know $\beta<\pi$. For such a $\beta$, it is not hard to check that under the constraint $2\alpha+\beta \leq 2\pi$, the cost function $\cos(\alpha)+\cos(\beta)+cos(\alpha+\beta)$ achieves its unique minimum when $2\alpha+\beta=2\pi$. Moreover, the cost function is a strictly decreasing function as $\alpha$ grows from $0$ to $\pi-\frac{\beta}{2}$. So if we shift $\tilde{\theta}_4$ counter-clockwise by a small amount $\delta>0$ and shift $\tilde{\theta}_1$ clockwise by the same small amount $\delta>0$, then as long as $2\alpha+\beta<2\pi$, this will strictly decrease the condition numbers simultaneously for $\{\tilde{\theta}_1, \tilde{\theta}_2, \tilde{\theta}_3\}$ and $\{\tilde{\theta}_2, \tilde{\theta}_3, \tilde{\theta}_4\}$.

 Since the cost functions for any other $3$ adjacent angles were strictly smaller than the original cost function of $\{\tilde{\theta}_1, \tilde{\theta}_2, \tilde{\theta}_3\}$, their cost functions will stay smaller than the new revised cost function of $\{\tilde{\theta}_1, \tilde{\theta}_2, \tilde{\theta}_3\}$. So we have just strictly decreased the maximum condition number, which is a contradiction to our assumption of an optimal design.
\end{proof}

By symmetry, in the same spirit, we have
\begin{lemma}
$\{\tilde{\theta}_1, \tilde{\theta}_2, \tilde{\theta}_3\}$ and $\{\tilde{\theta}_1, \tilde{\theta}_2, \tilde{\theta}_5\}$  can not be the only $2$ sets of $3$ adjacent angles which have the largest condition number.
\end{lemma}

 We can also prove:
\begin{lemma}
$\{\tilde{\theta}_1, \tilde{\theta}_2, \tilde{\theta}_3\}$ and $\{\tilde{\theta}_3, \tilde{\theta}_4, \tilde{\theta}_5\}$  can not be the only $2$ sets of $3$ adjacent angles which correspond to the maximum condition number.
\label{lemma:twoanglesetnotpossible}
\end{lemma}

\begin{proof}
Again, we prove by contradiction. Suppose $\{\tilde{\theta}_1, \tilde{\theta}_2, \tilde{\theta}_3\}$ and $\{\tilde{\theta}_3, \tilde{\theta}_4, \tilde{\theta}_5\}$ are the only two sets of $3$ adjacent angles that have the largest condition number.

We first assume that $\tilde{\theta}_5 -\tilde{\theta}_3 \neq \pi$.

We claim that if $\tilde{\theta}_5-\tilde{\theta}_3>\pi$, then $\tilde{\theta}_4=\frac{\tilde{\theta}_3+\tilde{\theta}_5}{2}$. This is because otherwise, we can shift $\tilde{\theta}_4$ toward the middle point $\frac{\tilde{\theta}_3+\tilde{\theta}_5}{2}$ by a sufficiently small amount, thus strictly decreasing the condition number for $\{\tilde{\theta}_3, \tilde{\theta}_4, \tilde{\theta}_5\}$. This will leave $\{\tilde{\theta}_1, \tilde{\theta}_2, \tilde{\theta}_3\}$ as the unique adjacent-$3$-set with the maximum condition number, which is not possible by Lemma \ref{lem:5singleoptimal}.

But if $\tilde{\theta}_5-\tilde{\theta}_3>\pi$, we must have $\tilde{\theta}_3-\tilde{\theta}_1<\pi$. However, then $\{\tilde{\theta}_1, \tilde{\theta}_2, \tilde{\theta}_3\}$ and $\{\tilde{\theta}_3, \tilde{\theta}_4, \tilde{\theta}_5\}$ can not have the same condition number. In fact, from analyzing the cost function, when $\tilde{\theta}_3-\tilde{\theta}_1<\pi$, $\tilde{\theta}_5-\tilde{\theta}_3>\pi$ and $\tilde{\theta}_4=\frac{\tilde{\theta}_3+\tilde{\theta}_5}{2}$,        $\{\tilde{\theta}_3, \tilde{\theta}_4, \tilde{\theta}_5\}$ has a strictly smaller condition number than $\{\tilde{\theta}_1, \tilde{\theta}_2, \tilde{\theta}_3\}$.

So when $\tilde{\theta}_5 -\tilde{\theta}_3 \neq \pi$,  we must have $\tilde{\theta}_5-\tilde{\theta}_3<\pi$ and, symmetrically, $\tilde{\theta}_3-\tilde{\theta}_1<\pi$. Then $\tilde{\theta}_4=\tilde{\theta}_3$ or $\tilde{\theta}_5$; $\tilde{\theta}_2=\tilde{\theta}_1$ or $\tilde{\theta}_3$; and
$\tilde{\theta}_3-\tilde{\theta}_1=\tilde{\theta}_5-\tilde{\theta}_3$.  This is because, if $\tilde{\theta}_4\neq \tilde{\theta}_3$ and $\tilde{\theta}_4\neq \tilde{\theta}_5$, we can always shift $\tilde{\theta}_4$ towards whatever is closer to $\tilde{\theta}_4$ among $\tilde{\theta}_3$ and $\tilde{\theta}_5$. This will strictly decrease the corresponding cost function, and leaving only one adjacent-$3$-angle set having the maximum condition number, which is not possible by Lemma \ref{lem:5singleoptimal}.

But then by increasing $\tilde{\theta}_3-\tilde{\theta}_1=\tilde{\theta}_5-\tilde{\theta}_3$ by a sufficiently small amount $\delta>0$, we will strictly decrease the condition numbers for $\{\tilde{\theta}_1, \tilde{\theta}_2, \tilde{\theta}_3\}$ and $\{\tilde{\theta}_3, \tilde{\theta}_4, \tilde{\theta}_5\}$. Since the cost functions for the other sets of $3$ adjacent angles were strictly smaller than the original cost function of $\{\tilde{\theta}_1, \tilde{\theta}_2, \tilde{\theta}_3\}$ and $\{\tilde{\theta}_3, \tilde{\theta}_4, \tilde{\theta}_5\}$, their cost functions will remain smaller than the new revised cost function of $\{\tilde{\theta}_3, \tilde{\theta}_4, \tilde{\theta}_5\}$. So we have just decreased the maximum condition number of the optimal design, which is not possible.

 We now consider the possibility that $\tilde{\theta}_5 - \tilde{\theta}_3 = \pi$. Because $\{\tilde{\theta}_3, \tilde{\theta}_4, \tilde{\theta}_5\}$ and $\{\tilde{\theta}_1, \tilde{\theta}_2, \tilde{\theta}_3\}$ both have the maximum condition number; and
  $(\tilde{\theta}_3-\tilde{\theta}_1)+(\tilde{\theta}_5-\tilde{\theta}_3)\leq 2\pi$, with the cost function $\cos(\alpha)+\cos(\beta)+\cos(\alpha+\beta)$ achieving the minimum $-1$ when $\alpha+\beta\leq \pi$  we must have $\tilde{\theta}_3-\tilde{\theta}_1=\pi$ too. Then $\tilde{\theta}_5$ and $\tilde{\theta}_1$ must be in the same position, and so $\{\tilde{\theta}_1, \tilde{\theta}_2, \tilde{\theta}_5\}$ must have a condition number no smaller than $\{\tilde{\theta}_1, \tilde{\theta}_2, \tilde{\theta}_3\}$ since $\cos(\tilde{\theta}_2-\tilde{\theta}_1)+\cos(\tilde{\theta}_1-\tilde{\theta}_5+2\pi)+\cos(\tilde{\theta}_2-\tilde{\theta}_5+2\pi)$ achieves its minimum $-1$ with $\tilde{\theta}_2=\pi$ when $\tilde{\theta}_2 \leq \pi$. This is contradictory to our assumption that $\{\tilde{\theta}_1, \tilde{\theta}_2, \tilde{\theta}_3\}$ and $\{\tilde{\theta}_3, \tilde{\theta}_4, \tilde{\theta}_5\}$ are the only $2$ sets of $3$ adjacent angles which have the maximum condition number.
 \end{proof}

By symmetry, we can also prove
\begin{lemma}
$\{\tilde{\theta}_1, \tilde{\theta}_2, \tilde{\theta}_3\}$ and $\{\tilde{\theta}_4, \tilde{\theta}_5,\tilde{\theta}_1\}$  can not be the only $2$ sets of $3$ adjacent angles which have the largest condition number.
\end{lemma}

\subsection{(At Least $4$ adjacent-$3$-angle Sets Giving the Maximum Condition Number) }

Now we consider the cases where more angle sets have the maximum condition number.
\begin{lemma}
$\{\tilde{\theta}_1, \tilde{\theta}_2, \tilde{\theta}_3\}$, $\{\tilde{\theta}_2, \tilde{\theta}_3, \tilde{\theta}_4\}$ and $\{\tilde{\theta}_3, \tilde{\theta}_4, \tilde{\theta}_5\}$ can not be the only $3$ sets of $3$ adjacent angles which have the largest condition number.
\end{lemma}

\begin{proof}
We prove by contradiction. Let us assume $\{\tilde{\theta}_1, \tilde{\theta}_2, \tilde{\theta}_3\}$, $\{\tilde{\theta}_2, \tilde{\theta}_3, \tilde{\theta}_4\}$ and $\{\tilde{\theta}_3, \tilde{\theta}_4, \tilde{\theta}_5\}$ are the only $3$ sets of $3$ adjacent angles which have the largest condition number. Apparently, $\tilde{\theta}_2-\tilde{\theta}_1=\tilde{\theta}_4-\tilde{\theta}_3=\alpha$  and $\tilde{\theta}_3-\tilde{\theta}_2=\tilde{\theta}_5-\tilde{\theta}_4=\beta$ for some $\alpha \geq 0$ and $\beta\geq 0$.

We must have $\alpha+\beta<\pi$. Otherwise, angle $\tilde{\theta}_5$ will be in the same position as $\tilde{\theta}_1$. But, as argued in Lemma \ref{lemma:twoanglesetnotpossible}, this implies $\{\tilde{\theta}_5, \tilde{\theta}_1, \tilde{\theta}_2\}$ can not have a smaller condition number than $\{\tilde{\theta}_1, \tilde{\theta}_2, \tilde{\theta}_3\}$, which is a contradiction to the assumption that $\{\tilde{\theta}_1, \tilde{\theta}_2, \tilde{\theta}_3\}$, $\{\tilde{\theta}_2, \tilde{\theta}_3, \tilde{\theta}_4\}$ and $\{\tilde{\theta}_3, \tilde{\theta}_4, \tilde{\theta}_5\}$ are the only $3$ sets of $3$ adjacent angles which have the largest condition number.

So we can always increase $\alpha$ and $\beta$ by a sufficiently small amount $\delta$ to decrease the condition number for $\{\tilde{\theta}_1, \tilde{\theta}_2, \tilde{\theta}_3\}$, $\{\tilde{\theta}_2, \tilde{\theta}_3, \tilde{\theta}_4\}$ and $\{\tilde{\theta}_3, \tilde{\theta}_4, \tilde{\theta}_5\}$. Since the cost functions for any other $3$ adjacent angles were strictly smaller than the original cost function of $\{\tilde{\theta}_1, \tilde{\theta}_2, \tilde{\theta}_3\}$, $\{\tilde{\theta}_2, \tilde{\theta}_3, \tilde{\theta}_4\}$ and $\{\tilde{\theta}_3, \tilde{\theta}_4, \tilde{\theta}_5\}$, their cost functions will remain smaller than the new revised cost function of $\tilde{\theta}_1$, $\tilde{\theta}_2$ and $\tilde{\theta}_3$. So we have just decreased the maximum condition number, which is a contradiction to our assumption.
\end{proof}

We also have:
\begin{lemma}
$\{\tilde{\theta}_1, \tilde{\theta}_2, \tilde{\theta}_3\}$, $\{\tilde{\theta}_2, \tilde{\theta}_3, \tilde{\theta}_4\}$ and $\{\tilde{\theta}_4, \tilde{\theta}_5, \tilde{\theta}_1\}$ can not be the only $3$ sets of $3$ adjacent angles which have the largest condition number.
\end{lemma}

\begin{proof}
Again, we prove by contradiction. Suppose that $\{\tilde{\theta}_1, \tilde{\theta}_2, \tilde{\theta}_3\}$, $\{\tilde{\theta}_2, \tilde{\theta}_3, \tilde{\theta}_4\}$ and $\{\tilde{\theta}_4, \tilde{\theta}_5, \tilde{\theta}_1\}$ are the only $3$ sets of $3$ adjacent angles which have the largest condition number.

Firstly, we assume that the counter-clockwise region between angle $\tilde{\theta}_4$ and angle $\tilde{\theta}_1$  is smaller than $\pi$. Then we know angle $\tilde{\theta}_5$  must be in the same position as angle $\tilde{\theta}_4$  or angle $\tilde{\theta}_1$. Otherwise, as we discussed earlier, we can always shift $\tilde{\theta}_5$ such that the cost function for $\{\tilde{\theta}_4, \tilde{\theta}_5, \tilde{\theta}_1\}$ is decreased, which will reduce us to the scenario in Lemma \ref{lemma:adjacenttwo}.

 Suppose $\tilde{\theta}_5$ is in the same position as angle $\tilde{\theta}_4$. From our assumption, the cost function for  $\{\tilde{\theta}_3, \tilde{\theta}_4, \tilde{\theta}_5\}$ is smaller than the cost function for $\{\tilde{\theta}_2, \tilde{\theta}_3, \tilde{\theta}_4\}$. This is not possible, because $\tilde{\theta}_4-\tilde{\theta}_3 <\pi$, and $\tilde{\theta}_5-\tilde{\theta}_4=0$ and the cost function for $\{\tilde{\theta}_2, \tilde{\theta}_3, \tilde{\theta}_4\}$ is maximized when $\tilde{\theta}_2=\tilde{\theta}_3$ under the condition $\tilde{\theta}_2\leq \tilde{\theta}_3$.

 By symmetry of $\tilde{\theta}_5$ with respect to $\tilde{\theta}_4$ and $\tilde{\theta}_1$, when $\tilde{\theta}_5$  is in the same position as $\tilde{\theta}_1$ , we also get a contradiction.

Secondly, we assume that the counter-clockwise region going from $\tilde{\theta}_4$ to $\tilde{\theta}_1$  is equal to $\pi$. In this case, the cost function for $\{\tilde{\theta}_5, \tilde{\theta}_1,\tilde{\theta}_4\}$ does not depend on the location of angle $\tilde{\theta}_5$. Since the cost function for $\{\tilde{\theta}_4, \tilde{\theta}_5,\tilde{\theta}_1\}$ is $-1$ and $\tilde{\theta}_4-\tilde{\theta}_1$, in order for  $\{\tilde{\theta}_1, \tilde{\theta}_2, \tilde{\theta}_3\}$, $\{\tilde{\theta}_2, \tilde{\theta}_3, \tilde{\theta}_4\}$ and $\{\tilde{\theta}_4, \tilde{\theta}_5, \tilde{\theta}_1\}$ to have the same cost function, we must have $\tilde{\theta}_2=\tilde{\theta}_1=0$ and $\tilde{\theta}_3=\tilde{\theta}_4=\pi$. This is contradictory to the assumption that $\{\tilde{\theta}_2, \tilde{\theta}_3,\tilde{\theta}_4\}$ has a larger cost function than $\{\tilde{\theta}_3, \tilde{\theta}_4,\tilde{\theta}_5\}$.

Thirdly, we assume that the counter-clockwise region going from $\tilde{\theta}_4$ to $\tilde{\theta}_1$ is larger than $\pi$. Similar to earlier analysis for the case that  , $\tilde{\theta}_5$ must at the middle point of the counter-clockwise region going from $\tilde{\theta}_4$ to $\tilde{\theta}_1$. However, we get a contradiction because the cost function for $\{\tilde{\theta}_4, \tilde{\theta}_5,\tilde{\theta}_1\}$ is no bigger than $-1$; while the cost function for $\{\tilde{\theta}_1, \tilde{\theta}_2,\tilde{\theta}_3\}$ is bigger than $-1$ since $\tilde{\theta}_3-\tilde{\theta}_1<\pi$.

So in summary, we have proven this lemma.
 \end{proof}

In the same spirit, we can prove
\begin{lemma}
$\{\tilde{\theta}_1, \tilde{\theta}_2, \tilde{\theta}_3\}$, $\{\tilde{\theta}_3, \tilde{\theta}_4, \tilde{\theta}_5\}$ and $\{\tilde{\theta}_5, \tilde{\theta}_1, \tilde{\theta}_2\}$ can not be the only $3$ sets of $3$ adjacent angles which have the largest condition number.
\end{lemma}

\begin{lemma}
$\{\tilde{\theta}_1, \tilde{\theta}_2, \tilde{\theta}_3\}$, $\{\tilde{\theta}_3, \tilde{\theta}_4, \tilde{\theta}_5\}$ and $\{\tilde{\theta}_4, \tilde{\theta}_5, \tilde{\theta}_1\}$ can not be the only $3$ sets of $3$ adjacent angles which have the largest condition number.
\end{lemma}

So the only left four possibilities are
\begin{itemize}

\item  $\{\tilde{\theta}_1, \tilde{\theta}_2, \tilde{\theta}_3\}$, $\{\tilde{\theta}_2, \tilde{\theta}_3, \tilde{\theta}_4\}$,  $\{\tilde{\theta}_3, \tilde{\theta}_4, \tilde{\theta}_5\}$,  and $\{\tilde{\theta}_4, \tilde{\theta}_5, \tilde{\theta}_1\}$ are the sets of $3$ adjacent angles which have the largest condition number.

\item $\{\tilde{\theta}_1, \tilde{\theta}_2, \tilde{\theta}_3\}$, $\{\tilde{\theta}_2, \tilde{\theta}_1, \tilde{\theta}_5\}$,  $\{\tilde{\theta}_1, \tilde{\theta}_5, \tilde{\theta}_4\}$,  and $\{\tilde{\theta}_5, \tilde{\theta}_4, \tilde{\theta}_3\}$ are the sets of $3$ adjacent angles which have the largest condition number.

\item $\{\tilde{\theta}_1, \tilde{\theta}_2, \tilde{\theta}_3\}$, $\{\tilde{\theta}_2, \tilde{\theta}_3, \tilde{\theta}_4\}$,  $\{\tilde{\theta}_3, \tilde{\theta}_4, \tilde{\theta}_5\}$,  and $\{\tilde{\theta}_5, \tilde{\theta}_1, \tilde{\theta}_2\}$ are the sets of $3$ adjacent angles which have the largest condition number.

\item $\{\tilde{\theta}_1, \tilde{\theta}_2, \tilde{\theta}_3\}$, $\{\tilde{\theta}_2, \tilde{\theta}_3, \tilde{\theta}_4\}$,  $\{\tilde{\theta}_4, \tilde{\theta}_5, \tilde{\theta}_1\}$,  and $\{\tilde{\theta}_5, \tilde{\theta}_1, \tilde{\theta}_2\}$ are the sets of $3$ adjacent angles which have the largest condition number.
\end{itemize}

These four cases are symmetric to each other, so we consider the first case and the conclusion carries over to the other three cases accordingly.

The first case implies that $\tilde{\theta}_2-\tilde{\theta}_1=\tilde{\theta}_4-\tilde{\theta}_3=\tilde{\theta}_1-\tilde{\theta}_5+2\pi=\alpha$ for some constant $\alpha \geq 0$; and $\tilde{\theta}_3-\tilde{\theta}_2=\tilde{\theta}_5-\tilde{\theta}_4=\beta$ for some constant $\beta\geq 0$.

 We note that $\{\tilde{\theta}_1, \tilde{\theta}_2, \tilde{\theta}_5\}$ are adjacent $3$ angles with $\alpha$ between $\tilde{\theta}_1$ and $\tilde{\theta}_2$; and $\alpha$ between $\{\tilde{\theta}_1$ and $\tilde{\theta}_5\}$. We also notice that $\alpha \geq \beta \geq 0$ (because $\alpha+\beta \leq \pi$ and $2\alpha+\beta \leq 2\pi$. Under these constraints, it is not hard to check that the cost function $\cos(\alpha)+\cos(\alpha)+\cos(2\alpha)$ for $\{\tilde{\theta}_5, \tilde{\theta}_1, \tilde{\theta}_2\}$ is bigger than the cost function $\cos(\alpha)+\cos(\beta)+\cos(\alpha+\beta)$ for $\{\tilde{\theta}_1, \tilde{\theta}_2, \tilde{\theta}_3\}$ if and only if  $\alpha \geq \beta$. ) and $3\alpha+2\beta=2\pi$. In order to minimize the largest condition number, we should make the cost function $\cos(\alpha)+\cos(\beta)+\cos(\alpha+\beta)$ as small as possible.

 Under the constraints that $\alpha \geq \beta \geq 0$ and $3\alpha+2\beta=2\pi$, we have $\frac{2\pi}{5}\leq \alpha \leq \frac{2\pi}{3}$. Within this range, the cost function $\cos(\alpha)+\cos(\beta)+\cos(\alpha+\beta)$ achieves its minimum when $\alpha=\beta=\frac{2\pi}{5}$. If $\alpha=\beta$, the cost function is equal to $2\cos(\frac{2\pi}{5})+\cos(\frac{4\pi}{5})\approx-0.1910$.

 So indeed the optimal solution is given by $\theta_i=\frac{\pi (i-1)}{N}$, $1 \leq i \leq  N$.
\end{proof}

\section{$K=3$, $N=7$}
One might think that the uniform distributed design is optimal for $N=7$

\begin{theorem}
Let $K=3$ and $N =7$ . Then $\theta_i=\frac{2\pi (i-1)}{N+1} \mod \pi$, $1 \leq i \leq  N$, minimizes the maximum condition number among all sub-matrices with $K=3$ columns.
\label{thm:K3N7}
\end{theorem}

\begin{proof}
%
 Among $\hat{\theta}_i=2\theta_i=\frac{4\pi (i-1)}{N+1} \mod 2\pi$, $1 \leq i \leq  N$, it is not hard to check that $3$ adjacent angles, denoted by $r_1$, $r_2$ and $r_3$ with $r_1=r_2$, give the maximum cost function
\begin{equation*}
 \cos(r_1-r_2)+\cos(r_1-r_3)+\cos(r_2-r_3)=2\cos(\frac{4\pi}{N+1})+1.
 \end{equation*}
This means the submatrix corresponding to such $3$ adjacent angles generate the maximum condition number among all possible $3$-column submatrices.

So in order to prove that $\hat{\theta}_i=2\theta_i=\frac{4\pi (i-1)}{N+1} \mod 2\pi$, $1 \leq i \leq N$, minimizes the maximum condition number among all $3$-column submatrices, it is enough to show $\hat{\theta}_i=2\theta_i=\frac{4\pi (i-1)}{N+1} \mod 2\pi$, $1 \leq i \leq  N$, minimizes the maximum condition number among all the adjacent-$3$-angle sets. Let us assume that $N$ angles $0\leq \tilde{\theta}=2\theta_i^* < 2\pi$, $1\leq i \leq N$ achieves the smallest maximum condition number among all the adjacent-$3$-angle sets. Without sacrificing generality, let $\tilde{\theta}_1=0$ and $\tilde{\theta}_i$ be arranged sequentially in a counter-clockwise order as $i$ goes from $1$ to $N$.

\begin{lemma}
In the optimal design $\tilde{\theta}_i$, $1\leq i \leq N$, the counter-clockwise region covered by any three adjacent angles is no bigger than $\pi$ for $N=7$; and smaller than $\pi$ for $N\geq 9$.  The only scenario where the counter-clockwise region covered by one adjacent-$3$-angle set is $\pi$ is when $N=7$ and the $7$ angles are respectively $0$, $0$, $\pi/2$, $\pi/2$, $\pi$,$\pi$ and $\frac{3\pi}{2}$ (up to rotations of these angles).
\label{lemma:lessthanpi}
\end{lemma}
\begin{proof}
Suppose instead in the optimal design, the counter-clockwise region covered by some three adjacent angles $r_1$, $r_2$ and $r_3$ is larger than $\pi$. Then there must exist $3$ adjacent angles for which the counter-clockwise region covered by them is smaller than $\frac{3\pi}{N-1}\leq \frac{\pi}{2}$ because the sum of the counter-clockwise regions covered by all the $3$ adjacent angles is $4\pi$ (Please see the proof in Theorem \ref{thm:K3NEven}). This means that the cost function for $r_1$, $r_2$ and $r_3$ is larger than $2\cos(\frac{3\pi}{N-1})+1$ (when $r_2$ is aligned with $r_1$ or $r_3$). Note that, $2\cos(\frac{3\pi}{N-1})+1$ is equal to the maximum cost function $2\cos(\frac{4\pi}{N+1})+1$ given by the to-be-proven optimal design $\hat{\theta}_i=2\theta_i=\frac{4\pi (i-1)}{N+1} \mod 2\pi$, $1 \leq i \leq  N$, when $N=7$; and is bigger than $2\cos(\frac{4\pi}{N+1})+1$ when $N\geq 9$. This is contradictory to our optimal design. So in the optimal design $\hat{\theta}_i=2\theta_i=\frac{4\pi (i-1)}{N+1} \mod 2\pi$, the counter-clockwise region covered by any three adjacent angles is no bigger than $\pi$ for $N\geq 7$.

When $N=7$ and the counter-clockwise region covered by one adjacent-$3$-angle set is $\pi$, the counter-clockwise region covered by each one of the other adjacent-$3$-angle sets is forced  to be $\frac{pi}{2}$, The only way for that to happen is that the $7$ angles $\tilde{\theta}_i$, $1\leq i \leq N$, are $0$, $0$, $\pi/2$, $\pi/2$, $\pi$,$\pi$ and $\frac{3\pi}{2}$ (up to rotations of these angles).

Note that the same argument can show that the counter-clockwise region covered by any three adjacent angles is smaller than $\pi$ for $N\geq 9$.
 \end{proof}

In the optimal design $\tilde{\theta}_i$, $1\leq i \leq N$, there are $N$ sets of $3$ adjacent angles, and we denote each set by its counter-clockwise central angle. For example, we denote the set of three angles $\{\tilde{\theta}_{(j-1)\mod N}, \tilde{\theta}_{j}, \tilde{\theta}_{(j+1)\mod N}\}$ for some $1\leq j \leq N$ as its central angle $\{\tilde{\theta}_{j}\}$. We assume that $\tilde{\theta}_1$, $\tilde{\theta}_2$ and $\tilde{\theta}_3$ are three angles which correspond to the largest condition number.

We now prove the following lemma:
\begin{lemma}
In the optimal design  $\tilde{\theta}_i$, $1\leq i \leq N$, $N\geq 7$, there do not exist $\geq 2$ consecutive adjacent-$3$-angle sets (for example $\{\tilde{\theta}_{j}\}$ and $\{\tilde{\theta}_{(j+1)\mod N}\}$ for some $1\leq j \leq N$) which have smaller condition numbers than the maximum condition number.
\label{lem:notwoconsecutivebigger}
\end{lemma}

\begin{proof}
We prove by contradiction. Suppose that for some $j$, $\{\tilde{\theta}_{j}\}$ and $\{\tilde{\theta}_{(j+1)\mod N}\}$ both have smaller condition numbers than the maximum condition number. We also assume that one of $\{\tilde{\theta}_{(j-1)\mod N}\}$ and $\{\tilde{\theta}_{(j+2)\mod N}\}$ is an adjacent-$3$-angle set corresponding to the maximum condition number. Notice that we can always find such a $j$ if there exist $\geq 2$ consecutive adjacent-$3$-angle sets which have smaller condition numbers than the maximum condition number.

By Lemma \ref{lemma:lessthanpi}, any adjacent angle widths $\gamma_1$ and $\gamma_2 $ satisfy $\gamma_1+\gamma_2 \leq \pi$. The cost function for $\cos(\gamma_1)+\cos(\gamma_2)+\cos(\gamma_1+\gamma_2)$ strictly decreases if we increase $\gamma_1$ and $\gamma_2$ simultaneously by a sufficiently small amount.

Suppose that $\{\tilde{\theta}_{j}\}$ spans two regions with counter-clockwise angle width $\alpha \geq 0$ and $\beta \geq 0$; and that $\{\tilde{\theta}_{(j+1)\mod N}\}$ spans two regions with counter-clockwise angle width $\beta$ and $\gamma$.

If $\beta>0$, we can always reduce $\beta$ by a sufficiently small enough amount and increase every region involved in all the adjacent-$3$-angle sets corresponding to the maximum condition number by an appropriate small amount such that the angle widths of the $N$ regions still sum up to $2\pi$. In this way, we have just strictly decreased the maximum condition number among all the adjacent-$3$-angle sets. This is contradictory to the optimal design assumption.

If $\beta=0$, since every two adjacent angles are no more than $\pi$ apart, $\{\tilde{\theta}_{(j-1)\mod N}\}$ has no bigger condition number than $\{\tilde{\theta}_{j}\}$ and $\{\tilde{\theta}_{(j+2)\mod N}\}$ has no bigger condition number than $\{\tilde{\theta}_{(j+1)\mod N}\}$. This is contradictory to the assumption that $\{\tilde{\theta}_{j}\}$ and $\{\tilde{\theta}_{(j+1)\mod N}\}$ both have smaller condition numbers than the maximum condition number; and one of $\{\tilde{\theta}_{(j-1)\mod N}\}$ and $\{\tilde{\theta}_{(j+2)\mod N}\}$ is an adjacent-$3$-angle set corresponding to the maximum condition number.
\end{proof}

\begin{lemma}
Let $N = 7$. Suppose $\tilde{\theta}_i$, $1\leq i \leq N$, is an optimal design which minimizes the maximum condition number among all adjacent-$3$-angle sets. Then there exists at most $1$ adjacent-$3$-angle set which has smaller condition number than the maximum condition number among all adjacent-$3$-angle sets.
\label{lem:atmostonesmaller}
\end{lemma}

\begin{proof}
We prove this lemma by contradiction.

Without loss of generality, suppose that $\{\tilde{\theta}_{1}\}$ has a condition number smaller than the maximum number and there exist $\geq 2$ adjacent-$3$-angle sets which have smaller condition numbers than the maximum condition number among all adjacent-$3$-angle sets. Then there must exist a sequence of consecutive angles, say $\tilde{\theta}_i$, $1\leq i \leq l$, for some $3 \leq l \leq N$, such that any adjacent-$3$-angle set $\{\tilde{\theta}_j\}$, $2\leq j \leq l-1$ has the maximum condition number while the first counter-clockwise adjacent-$3$-angle set $\{\tilde{\theta}_l\}$ and the first clockwise adjacent-$3$-angle set $\{\tilde{\theta}_{1}\}$ have smaller condition numbers than the maximum condition number.

Since $\{\tilde{\theta}_j\}$, $2\leq j \leq l-1$, have the \emph{equal} maximum condition number, the counter-clockwise regions between $\{\tilde{\theta}_j\}$, $1\leq j \leq l$ must alternate between $\alpha\geq 0$ and $\beta\geq 0$, where $\alpha+\beta \leq \pi$. Without loss of generality, we assume that $\alpha\geq \beta$. For now we also assume that $\alpha+\beta<\pi$. From \ref{lemma:lessthanpi}, when $\alpha+\beta=\pi$, only $1$ adjacent-$3$-angle set has a smaller than the maximum condition number.

We first consider the case where $l$ is an odd number,namely we have an even number of regions between angle $\tilde{\theta}_{1}$ and angle $\tilde{\theta}_{l}$. Since $\alpha+\beta<\pi$, we claim that $\beta$ must be equal to $0$. Suppose instead $\beta \neq 0$. Then we can shift the even-number-indexed angles $\{\tilde{\theta}_j\}$, $1\leq j \leq l$, counter-clockwise by a sufficiently small amount. This will strictly decrease the condition numbers for $\{\tilde{\theta}_j\}$, $2\leq j \leq l-1$. Since  $\{\tilde{\theta}_l\}$ and $\{\tilde{\theta}_{1}\}$ also have strictly smaller condition numbers than the maximum condition number, thus we have $\geq 2$ consecutive adjacent-$3$-angle sets which have the smaller condition number than the maximum condition number. This forms a contradiction by Lemma \ref{lem:notwoconsecutivebigger}.

So we must have $\beta=0$. However, this implies $\{\tilde{\theta}_{(l+1)\mod{N}}\}$ has a condition number no bigger than that of $\{\tilde{\theta}_l\}$. Thus we have two consecutive adjacent-$3$-angle sets $\{\tilde{\theta}_l\}$ and $\{\tilde{\theta}_{(l+1)\mod{N}}\}$ which have smaller condition numbers than the maximum condition number. This forms a contradiction by Lemma \ref{lem:notwoconsecutivebigger}.

We then consider the case where $l$ is an even number. Since $l\geq3$, such a number can only be $l=4$ or $l=6$.

When $l=4$, then $\{\tilde{\theta}_{6}\}$ must also have a smaller condition number than the maximum condition number. This is because,  from Lemma \ref{lem:notwoconsecutivebigger}, $\{\tilde{\theta}_{5}\}$ and $\{\tilde{\theta}_{7}\}$ can not have smaller condition numbers than the maximum condition number. Moreover, if $\{\tilde{\theta}_{6}\}$ also has the maximum condition number, $\{\tilde{\theta}_{j\mod N}\}$, $4\leq j \leq N+1$, are then consecutive angles such that $\{\tilde{\theta}_{5}\}$, $\{\tilde{\theta}_{6}\}$ and  $\{\tilde{\theta}_{7}\}$ all have the maximum condition numbers, which is not possible by our previous discussion of the cases when $l$ is an odd number.

 However, when $\{\tilde{\theta}_{6}\}$ has a smaller condition number than the maximum condition number, there are an even number (in fact, $2$,) of regions between angle $\tilde{\theta}_{4}$ and angle $\tilde{\theta}_{6}$, which is not possible by our previous discussion.

So in summary, the original assumption of $\geq 2$ adjacent-$3$-angle sets having larger than maximum condition number cannot hold. There exists at most $1$ adjacent-$3$-angle set which has smaller condition number than the maximum condition number.

\end{proof}

If every adjacent-$3$-angle set has the same condition number as the maximum condition number, then the region between every angle must be equal. So the cost function for the maximum condition number should be $2\cos(\frac{2\pi}{7})+\cos(\frac{4\pi}{7})$.

If there is exactly $1$ adjacent-$3$-angle set which has a smaller condition number than the maximum condition number, and $\{\tilde{\theta}_{1}\}$ is the unique adjacent-$3$-angle set that has the smallest condition number, then $\tilde{\theta}_{i}$, $1\leq i \leq 7$, can be respectively denoted by $0$, $\alpha$, $\alpha+\beta$, $2\alpha+\beta$, $2(\alpha+\beta)$, $3\alpha+2\beta$, and $3\alpha+3\beta$, where $\alpha\geq 0$, $\beta \geq 0$ and $\alpha> \beta$ and $4\alpha+3\beta=2\pi$. The cost function for the maximum condition number  $\cos(\alpha)+\cos(\beta)+\cos(\alpha+\beta)$ is thus minimized when $\beta=0$ and $\alpha=\frac{2\pi}{4}$ for $\alpha\geq 0$, $\beta \geq 0$ and $\alpha> \beta$ and $4\alpha+3\beta=2\pi$. This cost function is smaller than the cost function of $2\cos(\frac{2\pi}{7})+\cos(\frac{4\pi}{7})$, so $\theta_i=\frac{2\pi (i-1)}{N+1} \mod 2\pi$, $1 \leq i \leq N$ is indeed the optimal design.
\end{proof}

\section{$K=3$, $N\geq9$ is an Odd Number}

\begin{theorem}
Let $K=3$ and $N \geq 9$ be an odd number. Then the set of angles $\theta_i=\frac{2\pi (i-1)}{N+1} \mod \pi$, $1 \leq i \leq  N$, minimizes the maximum condition number among all sub-matrices with $K=3$ columns.
\label{thm:K3Nodd}
\end{theorem}

\begin{proof}
The proof of this theorem follows the proof of Theorem \ref{thm:K3N7}. The complication compared with Theorem \ref{thm:K3N7} comes from the fact that we need to prove the following lemma instead of Lemma \ref{lem:atmostonesmaller}.

\begin{lemma}
Let us take $N \geq 9$. Suppose that $\tilde{\theta}_i$, $1\leq i \leq N$, is an optimal design which minimizes the maximum condition number among all adjacent-$3$-angle sets. Then there exists at most $1$ adjacent-$3$-angle set which has a smaller condition number than the maximum condition number among all adjacent-$3$-angle sets.
\label{lem:atmostonesmallerN9}
\end{lemma}

\begin{proof}
We prove this lemma by contradiction.

Suppose that there exists $\geq 2$ adjacent-$3$-angle sets which have smaller condition numbers than the maximum condition number among all adjacent-$3$-angle sets. From Lemma \ref{lem:notwoconsecutivebigger}, we can always partition the $N$ angles into distinct blocks by using $\tilde{\theta}_j$'s  with $\{\tilde{\theta}_j\}$ having a strictly smaller condition number than the maximum condition number as the boundary angles between different blocks. From Lemma \ref{lem:notwoconsecutivebigger}, there must exist at least one angle between two boundary angles. 
%
Without loss of generality, suppose $\tilde{\theta}_1$ and $\tilde{\theta}_l$, $3\leq l \leq N$, are two neighboring boundary angles. Since $\{\tilde{\theta}_j\}$, $2\leq j \leq l-1$, have the \emph{equal} maximum condition number, the counter-clockwise regions between $\{\tilde{\theta}_j\}$, $1\leq j \leq l$ must alternate between $\alpha\geq 0$ and $\beta\geq 0$, where $\alpha+\beta < \pi$ according to Lemma \ref{lemma:lessthanpi}.

We first consider the case when $l$ is an odd number, namely we have an even number of regions between angle $\tilde{\theta}_{1}$ and angle $\tilde{\theta}_{l}$. Without loss of generality, we assume that $\alpha\geq \beta$ when $l$ is an odd number. Since $\alpha+\beta<\pi$, from the same reasoning as in the proof of Lemma \ref{lem:atmostonesmaller}\, we know this is not possible. %
%

We then consider the case when $l$ is an even number. If $l$ is an even number, we divide into two scenarios: $\alpha \geq \beta$ or $\alpha \leq \beta$.

If $\alpha \leq \beta$, we can simultaneously shift the even-numbered angles $\tilde{\theta}_j$, $j=2, 4, ..., l-2$, clockwise by the same sufficiently small angle $\delta>0$. Note that this shift will not increase the maximum condition number if $\delta$ is sufficiently small. However, this will create two consecutive adjacent-$3$-angle sets $\{\tilde{\theta}_2\}$ and $\{\tilde{\theta}_1\}$ which have smaller condition numbers than the maximum condition number.  According to Lemma \ref{lem:notwoconsecutivebigger}, this is contradictory to our assumption of an optimal design.

We now assume $\alpha \geq \beta$ and the number of regions in each block is an odd number. Consider two neighboring blocks separated by a single angle $j$ such that $\{\tilde{\theta}_j\}$ is an adjacent-3-angle set which has a smaller condition number than the maximum condition number.  Suppose that the second block is in the clockwise direction of the first block. The counter-clockwise region in the first block alternates between $\alpha$ and $\beta$; the counter-clockwise region in the second block alternates between $\alpha_1$ and $\beta_1$ with $\alpha_1\geq \beta_1$ (otherwise we are done by the discussion in last paragraph). Since the adjacent-$3$-angle sets inside each block have the maximum condition number, without loss of generality, we have $\alpha_1 \leq \alpha$, and $\beta_1 \geq \beta$.
If we change the regions of the $2$-nd block to be $\beta_1$, $\alpha_1$, $\beta_1$, $\alpha_1$, ..., $\alpha_1$, and $\alpha_1$. Since $\alpha_1\leq \alpha$ and $\beta_1\geq\beta$, in this change, we do not increase the condition number of $\{\tilde{\theta}_j\}$. It is not hard to check that as long as $\alpha_1+\beta_1< \pi$, the cost function $\cos(\alpha_1)+\cos(\alpha_1)+\cos(2\alpha_1)$ is smaller than the cost function $\cos(\alpha_1)+\cos(\beta_1)+\cos(\alpha_1+\beta_1)$. So in this change, we do not increase the maximum condition number among adjacent-$3$-angle sets, while creating two consecutive adjacent-$3$-angle sets at the clockwise end of the second block, which is a contradiction from Lemma \ref{lem:notwoconsecutivebigger}.

So in summary, there exists at most $1$ adjacent-$3$-angle set which has smaller condition number than the maximum condition number.

\end{proof}

So in the optimal design, the angles must alternate like $\alpha$, $\beta$, ..., $\alpha$, $\beta$, $\alpha$, where $\alpha \geq \beta$, and $\frac{N+1}{2}\alpha+\frac{N-1}{2}\beta=2\pi$. For $N\geq9$, the optimal angle allocation for $\alpha$ is $\frac{4\pi}{N+1}$ and $\beta=0$.

\end{proof}

\section{Conclusion and Future Work}
We propose the problem designing optimal $M \times N$ ($M \leq  N$) sensing matrices which minimize the maximum condition number of all the submatrices of $K$ columns.  Such matrices minimize the worst-case estimation errors when only $K$ sensors out of $N$ sensors are available for sensing at a given time.   When $M=2$ and $K=3$, for an arbitrary $N\geq3$, we derive the optimal matrices which minimize the maximum condition number of all the submatrices of $K$ columns. It is interesting that minimizing the maximum coherence between columns does not always guarantee minimizing the maximum condition number.

\end{document}